%% file: Allerton.tex
\documentclass[10pt,conference]{IEEEtran}
\usepackage{epsf,psfrag,amssymb,amsfonts,cite,amsmath}

\input macro

\def\scalefig#1{\epsfxsize #1\textwidth}

\date{April 27, 2008}

\begin{document}

\title{\bf \LARGE Capacity Bounds for Two-Hop Interference Networks}
\author{\authorblockN{Yi Cao and Biao Chen}
\authorblockA{Department of EECS, Syracuse University, Syracuse, NY 13244\\
Email: ycao01@syr.edu, bichen@syr.edu}} \maketitle
\begin{abstract}
This paper considers a two-hop interference network, where two
users transmit independent messages to their respective receivers
with the help of two relay nodes. The transmitters do {\em not}
have direct links to the receivers; instead, two relay nodes serve
as intermediaries  between the transmitters and receivers. Each
hop, one from the transmitters to the relays and the other from
the relays to the receivers, is modeled as a Gaussian interference
channel, thus the network is essentially a cascade of two
interference channels. For this network, achievable symmetric
rates for different parameter regimes under decode-and-forward
relaying and amplify-and-forward relaying are proposed and the
corresponding coding schemes are carefully studied. Numerical
results are also provided.
\end{abstract}

\section{Introduction}

The wireless mesh networks are being extensively studied recently
due to their potential to improve the performance and throughput
of the cellular networks by borrowing the features of ad-hoc
networks \cite{Akyildiz:05CN}. The two-hop interference network
was recently proposed to model the mesh network from an
information theoretic perspective \cite{Simeone_etal:07Allerton}.
The model is in essence a cascade of two interference channels:
the transmitters communicate to two relay nodes through an
interference channel and the two relay nodes communicate to the
two receivers through another interference channel.

In \cite{Simeone_etal:07Allerton}, the authors studied the
achievable region for the model where the relays apply
decode-and-forward scheme. For the interference channel in the
first hop, since the messages of the two users are independent,
the largest achievable region to date was proposed by Han and
Kobayashi \cite{Han&Kobayashi:81IT}. The basic idea is for each
user to split their message into two parts: the private message,
which is only to be decoded by the intended receiver, and the
common message, which is to be decoded by both receivers. Although
the unintended user's common message is discarded by the receivers
in the classic interference channel model,
\cite{Simeone_etal:07Allerton} made use of this common message at
the two relay nodes as knowledge of them can help boost the rate
in the second hop through cooperative transmission. In
\cite{Simeone_etal:07Allerton}, the authors proposed the
superposition coding scheme for each relay node to transmit not
only the intended user's private and common messages but also the
other user's common message, in order to obtain the coherent
combining gain of the common message at the intended receiver.

\cite{Thejaswi_etal:07Allerton} also considered the two-hop
interference network model. Instead of considering the end-to-end
transmission rate, the authors focused on the the second hop and
explored the possibilities for the two relays to utilize the
common message from the unintended user
 and proposed multiple transmission schemes, such as
MIMO broadcast strategy, dirty paper coding, beamforming, and
further rate splitting.

However, both \cite{Simeone_etal:07Allerton} and
\cite{Thejaswi_etal:07Allerton} only considered decode-and-forward
relaying and focused on the weak interference case for both hops,
i.e., the interference link gain is less than the direct link
gain. In this paper, we study the model under various parameter
regimes using decode-and-forward relaying as well as
amplify-and-forward relaying. \cite{Simeone_etal:07Allerton} and
\cite{Thejaswi_etal:07Allerton} also suggested that, if the
interference channel in the first hop has strong interference
(interference link gain greater than direct link gain), by the
standard results for the classic interference channel, it is
optimal for the two relays to decode both users messages. Contrary
to this, we will show in later sections that this approach can be
easily outperformed by switching the roles of the two relays which
essentially converts the strong interference channels to weak
interference channels. For amplify and forward, we demonstrate
that the end-to-end rate may exceed the naive use of cut-set bound
which applies to only the decode and forward approach.

The rest of the paper is organized as follows. In section II, we
introduce the model for the two-hop interference network. In
section III, we focus on the end-to-end transmission rate and
analyze the decode-and-forward relaying scheme for the network
under different parameter regimes. In section IV, we analyze the
amplify-and-forward scheme under various parameter regimes.
Section V provides numerical examples to compare various proposed
coding schemes. Concluding remarks are given in section VI.

\section{Channel Model}
The standard two-hop interference network is a cascade of two
interference channels with direct transmission link coefficient
equal to $1$, as shown in Fig. \ref{fig:standard twohop}.
\begin{figure}[htb]
\centerline{
\begin{psfrags}
\psfrag{W1}[l]{$W_1$}\psfrag{W2}[l]{$W_2$}
\psfrag{W11}[l]{$\hat{W}_1$}\psfrag{W22}[l]{$\hat{W}_2$}
\psfrag{X1}[l]{$X_1$} \psfrag{X2}[l]{$X_2$} \psfrag{X3}[l]{$X_3$}
\psfrag{X4}[l]{$X_4$} \psfrag{Y1}[l]{$Y_1$} \psfrag{Y2}[l]{$Y_2$}
\psfrag{Y3}[l]{$Y_3$} \psfrag{Y4}[l]{$Y_4$} \psfrag{Z1}[l]{$Z_1$}
\psfrag{Z2}[l]{$Z_2$} \psfrag{Z3}[l]{$Z_3$} \psfrag{Z4}[l]{$Z_4$}
\psfrag{T1}[l]{$T1$}\psfrag{T2}[l]{$T2$}
\psfrag{R1}[l]{$R1$}\psfrag{R2}[l]{$R2$}
\psfrag{D1}[l]{$D1$}\psfrag{D2}[l]{$D2$}
\psfrag{h1}[l]{$1$}\psfrag{h2}[l]{$a_2$}
\psfrag{h3}[l]{$a_1$}\psfrag{h4}[l]{$1$}
\psfrag{h5}[l]{$1$}\psfrag{h6}[l]{$b_2$}
\psfrag{h7}[l]{$b_1$}\psfrag{h8}[l]{$1$}
\scalefig{.50}\epsfbox{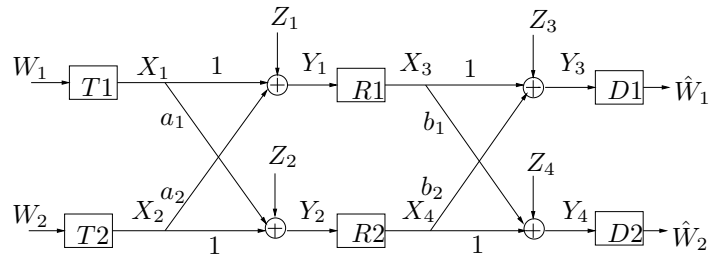}
\end{psfrags}
} \caption{\label{fig:standard twohop}Two-hop interference network
in standard form}
\end{figure}

Transmitter 1 ($T_1$) has message
$W_1\in\{1,2,\cdot\cdot\cdot,2^{nR_1}\}$ to be transmitted to
destination $D_1$ and transmitter 2 ($T_2$) has message
$W_2\in\{1,2,\cdot\cdot\cdot,2^{nR_2}\}$ to be transmitted to
destination $D_2$. $a_1$, $a_2$, $b_1$ and $b_2$ are fixed
positive numbers, $Z_1$, $Z_2$, $Z_3$ and $Z_4$ are independent
Gaussian distributed variables with zero mean and unit variance.
The average power constraints for the input signals $X_1$, $X_2$,
$X_3$ and $X_4$ are $P_{11}$, $P_{12}$, $P_{21}$ and $P_{22}$,
respectively.

In order to simplify the analysis of this complicated channel
model and better compare our results with the existing ones, we
follow the convention of \cite{Simeone_etal:07Allerton} and
\cite{Thejaswi_etal:07Allerton} by only considering the symmetric
interference channels, i.e., \bqa
a_1=a_2&\triangleq& a\label{eq:model5}\\
b_1=b_2&\triangleq& b\\
P_{11}=P_{12}&\triangleq& P_1\\
P_{21}=P_{22}&\triangleq& P_2\label{eq:model8} \eqa In addition,
we focus primarily on the symmetric rate, i.e., the case with
$R_1=R_2$.
\section{Decode and Forward}
In this section, we propose capacity bounds for the two-hop
interference network in various parameter regimes using
decode-and-forward relaying. Under the full duplex condition, the
transmission is conducted across a large number of blocks. In each
block, the relays receive the new messages of the current block
from the transmitters, and transmit the information of the
previous block to the desitnation. We assume the number of blocks
is large enough to ignore the penalty incurred in the first and
the last blocks.

\subsection{$0<a<1, 0<b<1$}\label{subsection:case1}
In \cite{Simeone_etal:07Allerton}, the authors proposed achievable
transmission rates for the case that both hops have weak
interference, i.e., $a<1$ and $b<1$. Specifically, they applied
Han-Kobayashi's scheme to the first hop by splitting each user's
message into two parts, namely, $W_1$ into private message
$W_{1p}\in \{1, \cdot\cdot\cdot, 2^{nR_{1p}}\}$ and common message
$W_{1c}\in \{1, \cdot\cdot\cdot, 2^{nR_{1c}}\}$ and $W_2$ into
private message $W_{2p}\in \{1, \cdot\cdot\cdot, 2^{nR_{2p}}\}$
and common message $W_{2c}\in \{1, \cdot\cdot\cdot,
2^{nR_{2c}}\}$. Each relay not only decodes the private and common
messages from the intended user, but also decodes the common
message from the other user. Since the Han-Kobayashi region is
based on simultaneous decoding of the three messages (1 private
message and 2 common messages), which is very complicated to
compute, \cite{Thejaswi_etal:07Allerton} simplified it by
proposing sequential decoding: each relay first decodes the two
common messages, subtract them out, then decode the private
message. By restricting the analysis to the symmetric rate
\cite{Thejaswi_etal:07Allerton}, i.e., $R_{1p}=R_{2p}=R_p^{(1)}$,
$R_{1c}=R_{2c}=R_c^{(1)}$, we have achievable rates in the first
hop \bqa R_p^{(1)}\!\!\!\!\!&=&\!\!\!\!\!\gamma\left(\frac{\alpha
P_1}{1+a^2\alpha P_1}\right)\label{eq:1}\\
R_c^{(1)}\!\!\!\!\!&=&\!\!\!\!\!\min\left\{\gamma\left(\frac{a^2\bar{\alpha}P_1}{\sigma_1^2}\right),
\frac{1}{2}\gamma\left(\frac{(1+a^2)\bar{\alpha}P_1}{\sigma_1^2}\right)\right\}\label{eq:2}
\eqa where $\alpha P_1$ is the power allocated to the private
message and $\bar{\alpha}P_1=(1-\alpha)P_1$ is the power allocated
to the common message. $\sigma_1^2=1+(1+a^2)\alpha P_1$.
$\gamma(x)$ is defined as $\frac{1}{2}\log(1+x)$. The superscript
``(1)" denotes the first hop. (\ref{eq:2}) is from the capacity
region of the MAC channel consisting of the two common messages,
treating the private messages as noise; (\ref{eq:1}) is the
decoding of the private message treating the other private message
as noise.

For the second hop, \cite{Simeone_etal:07Allerton} proposed
superposition scheme at the two relays such that coherent
combining can be achieved at the destinations. This scheme was
outperformed by the dirty paper coding (DPC) scheme proposed in
\cite{Thejaswi_etal:07Allerton} for the very weak interference
case, i.e., when $b$ is very small. The idea is for the two relays
to encode one of the common messages using DPC, thus treating the
other common message as known interference. Therefore, this known
interference will not affect the unintended destination. However,
due to the nonlinearity of the DPC, the dirty paper decoded common
message cannot be subtracted out. Thus,
\cite{Thejaswi_etal:07Allerton} also suggested to dirty paper code
the private message treating both common messages as known
interference. Besides, the common message that is treated as known
interference is decoded at its intended destination by treating
the other common message (dirty paper coded) as well as the two
private messages as noise. Since either common message can be
dirty paper coded against the other common message, there are two
transmission modes and one should time share between them to
maximize the sum rate \cite{Thejaswi_etal:07Allerton}. Again, by
only considering the symmetric rates, the achievable rates under
the DPC scheme for the second hop are
\cite{Thejaswi_etal:07Allerton} \bqa R_{p,
DPC}^{(2)}&=&\gamma\left(\frac{\beta
P_2}{1+b^2\beta P_2}\right)\label{eq:3}\\
R_{c,DPC}^{(2)}&=&\frac{1}{2}\gamma\left(\frac{(1-b^2)^2\bar{\beta}^2P_2^2}{\sigma_2^4}+\frac{2(1+b^2)\bar{\beta}P_2}{\sigma_2^2}\right)\label{eq:4}
\eqa where $\beta P_2$ is the power allocated to the private
message, $\sigma_2^2=1+(1+b^2)\beta P_2$ since the private
messages from both users are treated as noise when decoding common
messages. (\ref{eq:3}) is decoding the private message treating
the other user's private message as noise, since the effect of the
two common messages disappears due to the DPC; (\ref{eq:4}) is
from the optimization problem which maximizes the sum rate of the
two common messages.

In the DPC scheme, the common message that is treated as known
interference is decoded by its intended receiver treating the
other user's common message and private messages as noise.
However, when the interference link of the second hop gets
stronger, i.e., $b$ gets larger, the interference incurred by the
common message and private message from the other user may be too
strong to be treated as noise. Therefore, it may be beneficial for
the receivers to decode the common message and even the private
message from the other user, like in the strong interference
channel, whose capacity is that of the compound MAC. To make the
coding scheme more general, we do not let the receivers decode all
the private messages. Instead, we further split the private
message $W_{1p}$ from the first hop into two parts,
$W_{1pp}\in\{1,2,\cdot\cdot\cdot,2^{nR_{1pp}}\}$ and
$W_{1pc}\in\{1,2,\cdot\cdot\cdot,2^{nR_{1pc}}\}$, where $W_{1pp}$
is the sub-private message only decoded at the intended receiver,
and $W_{1pc}$ is the sub-common message decoded at both receivers.
The private message $W_{2p}$ is split in the same fashion into
$W_{2pp}$ and $W_{2pc}$. There are five messages (two common
messages, two sub-common messages and one sub-private message) to
be decoded by each receiver, which yields very complex expression
for the rate region if we use simultaneous decoding. Instead, we
will adopt sequential decoding and fix the decoding order as
follows: first, simultaneously decode the two common messages
$W_{1c}$ and $W_{2c}$, subtract them out; second, simultaneously
decode the two sub-common messages $W_{1pc}$ and $W_{2pc}$,
subtract them out; third, decode the sub-private message $W_{1pp}$
by receiver 1 (or $W_{2pp}$ by receiver 2). Consequently, the
symmetric achievable rate region is \bqa
R_c&\leq&\gamma\left(\frac{(\sqrt{P_{c1}}+b\sqrt{P_{c2}})^2}{1+(1+b^2)P_p}\right)\\
R_c&\leq&\gamma\left(\frac{(\sqrt{P_{c2}}+b\sqrt{P_{c1}})^2}{1+(1+b^2)P_p}\right)\\
2R_c&\leq&\gamma\left(\frac{(\sqrt{P_{c1}}+b\sqrt{P_{c2}})^2+(\sqrt{P_{c2}}+b\sqrt{P_{c1}})^2}{1+(1+b^2)P_p}\right)\\
R_{pc}&\leq&\gamma\left(\frac{P_{pc}}{1+(1+b^2)P_{pp}}\right)\\
R_{pc}&\leq&\gamma\left(\frac{b^2P_{pc}}{1+(1+b^2)P_{pp}}\right)\\
2R_{pc}&\leq&\gamma\left(\frac{(1+b^2)P_{pc}}{1+(1+b^2)P_{pp}}\right)\\
R_{pp}&\leq&\gamma\left(\frac{P_{pp}}{1+b^2P_{pp}}\right)
 \eqa
where power $P_p$ is allocated to the private message, $P_{c1}$ is
allocated to the intended common message, $P_{c2}$ is allocated to
the interfering common message, and $P_p+P_{c1}+P_{c2}=P_2$. Also,
$P_{pc}$ is for the sub-common message and $P_{pp}$ is for the
sub-private message and $P_{pc}+P_{pp}=P_p$. If we fix $P_p$ and
maximize $R_c$ under $P_{c1}+P_{c2}\leq P_2-P_p$, the optimal
$R_c^*=\frac{1}{2}\gamma\left(\frac{(1+b)^2(P_2-P_p)}{1+(1+b^2)P_p}\right)$
is achieved when $P_{c1}=P_{c2}=\frac{1}{2}(P_2-P_p)$
\cite{Thejaswi_etal:07Allerton}.  Therefore, the symmetric rates
for the second hop under the MAC scheme is \bqa\nn
R_{p,MAC}^{(2)}\!\!\!\!&=&\!\!\!\!\max_{\alpha}\left\{\min\left[\gamma\left(\frac{b^2\bar{\alpha}\beta
P_2}{\sigma_3^2}\right),
\frac{1}{2}\gamma\left(\frac{(1+b^2)\bar{\alpha}\beta
P_2}{\sigma_3^2}\right)\right]\right.\\
&&+\left.\gamma\left(\frac{\alpha\beta P_2}{1+b^2\alpha\beta
P_2}\right)\right\}\label{eq:3(2)}\\
R_{c,MAC}^{(2)}\!\!\!\!&=&\!\!\!\!\frac{1}{2}\gamma\left(\frac{(1+b)^2\bar{\beta}P_2}{1+(1+b^2)\beta
P_2}\right)\label{eq:4(2)}
 \eqa
where $\sigma_3^2=1+(1+b^2)\alpha\beta P_2$ and $\alpha, \beta \in
[0,1]$.

This scheme is more general than the cooperative transmission
scheme in \cite{Simeone_etal:07Allerton} in that we further split
the first hop's private messages into two parts in the second hop.
This scheme is similar to the ``layered coding with beamforming"
scheme in \cite{Thejaswi_etal:07Allerton}, with the difference
that we only consider the coherent beamforming here and disregard
the zero forcing beamforming scheme which proves to be always
worse than the DPC scheme.

\begin{theorem}\label{thm:1}
The achievable symmetric rate ($R_1=R_2=R$) for the symmetric
interference network is the solution to the following optimization
problem: \bqa R&=&\max_{\alpha,\beta\in [0,1]} R_p+R_c\\
&&\mbox{s.t.} (R_p,R_c)\in \Rmat(R_p^{(1)},R_c^{(1)})\cap
\Rmat(R_p^{(2)},R_c^{(2)})\eqa where $R_p^{(1)}$ and $R_c^{(1)}$
are given in (\ref{eq:1})-(\ref{eq:2}).
$\Rmat(R_p^{(2)},R_c^{(2)})$ is defined as the convex closure of
the union of $\Rmat(R_{p,DPC}^{(2)}, R_{c,DPC}^{(2)})$ and
$\Rmat(R_{p,MAC}^{(2)}, R_{c,MAC}^{(2)})$, where $R_{p,DPC}^{(2)}$
and $R_{c,DPC}^{(2)}$ are given in (\ref{eq:3})-(\ref{eq:4}), and
$R_{p,MAC}^{(2)}$ and $R_{c,MAC}^{(2)}$ are given in
(\ref{eq:3(2)})-(\ref{eq:4(2)}).
\end{theorem}

\subsection{$a>1,b>1$}\label{subsection:case2}
If the first hop has strong interference, i.e., $a>1$, both
\cite{Simeone_etal:07Allerton} and \cite{Thejaswi_etal:07Allerton}
let both relays decode both users' messages in the first hop, as
this is the optimal scheme for interference channels with strong
interference. Using this scheme, for the symmetric rates
$(R_1=R_2=R^{(1)})$, we have \bqa R^{(1)}&\leq& \gamma(P_1)\\
R^{(1)}&\leq& \gamma(a^2P_1)\\
R^{(1)}+R^{(1)}&\leq& \gamma(P_1+a^2P_1) \eqa Thus, \bqa
R^{(1)}=\min\left(\gamma(P_1),
\frac{1}{2}\gamma((1+a^2)P_1)\right) \eqa In other words, for very
strong interference case $a^2\geq 1+P_1$, $R^{(1)}=\gamma(P_1)$;
for $1<a^2<1+P_1$, $R^{(1)}=\frac{1}{2}\gamma((1+a^2)P_1)$.

After the first hop, since both relays have knowledge of both
users' messages, the second hop reduces to the Gaussian vector
broadcast channel with per antenna power constraint, for which we
know the DPC scheme is optimal. By time sharing between the two
DPC modes and maximizing the sum rate, we obtain the achievable
symmetric rate for the second hop \bqa
R^{(2)}=\frac{1}{2}\gamma((b^2-1)^2P_2^2+2P_2(1+b^2)). \eqa
Therefore the achievable rate for the entire network is \bqa
R=\min\{R^{(1)}, R^{(2)}\}. \label{eq:5} \eqa

The above analysis seems to be a natural way to deal with the
strong interference case, and for each hop, the transmission
scheme is optimal. However, optimality in each hop does not
guarantee optimality of the entire network. Indeed, for the entire
system, the combination of the two optimal schemes is no longer
optimal. An easy way to outperform the above scheme is to switch
the role of the two relays. Specifically, we make relay $R_2$ as
the ``intended" relay for the first user $T_1$, and relay $R_1$ as
the intended relay for the second user $T_2$. In this way, the
first hop is converted into an interference channel with weak
interference. Consequently, the second hop is converted into
another weak interference channel as shown in
Fig.\ref{fig:transform}. After some simple scaling, this two-hop
network becomes \bqa Y_1^{'}&=&X_1+\frac{1}{a}X_2+Z_1^{'}\\
Y_2^{'}&=&\frac{1}{a}X_1+X_2+Z_2^{'}\\
Y_3^{'}&=&X_3+\frac{1}{b}X_4+Z_3^{'}\\
Y_4^{'}&=&\frac{1}{b}X_3+X_4+Z_4^{'} \eqa where $Z_1^{'},
Z_2^{'}\sim N(0, 1/a^2)$, $Z_3^{'}, Z_4^{'}\sim N(0, 1/b^2)$ are
independent.
\begin{figure}[htb]
\centerline{
\begin{psfrags}
\psfrag{W1}[l]{$W_1$}\psfrag{W2}[l]{$W_2$}
\psfrag{W11}[l]{$\hat{W}_1$}\psfrag{W22}[l]{$\hat{W}_2$}
\psfrag{X1}[l]{$X_1$} \psfrag{X2}[l]{$X_2$} \psfrag{X3}[l]{$X_4$}
\psfrag{X4}[l]{$X_3$} \psfrag{Y1}[l]{$Y_2$} \psfrag{Y2}[l]{$Y_1$}
\psfrag{Y3}[l]{$Y_3$} \psfrag{Y4}[l]{$Y_4$} \psfrag{Z1}[l]{$Z_2$}
\psfrag{Z2}[l]{$Z_1$} \psfrag{Z3}[l]{$Z_3$} \psfrag{Z4}[l]{$Z_4$}
\psfrag{T1}[l]{$T1$}\psfrag{T2}[l]{$T2$}
\psfrag{R1}[l]{$R2$}\psfrag{R2}[l]{$R1$}
\psfrag{D1}[l]{$D1$}\psfrag{D2}[l]{$D2$}
\psfrag{h1}[l]{$a$}\psfrag{h2}[l]{$1$}
\psfrag{h3}[l]{$1$}\psfrag{h4}[l]{$a$}
\psfrag{h5}[l]{$b$}\psfrag{h6}[l]{$1$}
\psfrag{h7}[l]{$1$}\psfrag{h8}[l]{$b$}
\scalefig{.50}\epsfbox{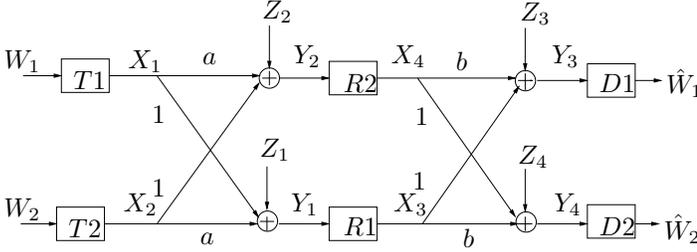}
\end{psfrags}
} \caption{\label{fig:transform}Two-hop interference network
transformation}
\end{figure}

Therefore, this strong interference two-hop network reduces to
case \ref{subsection:case1} where both hops are weak interference
channels. Using Han-Kobayashi scheme in the first hop and
combining DPC and MAC in the second hop, and going through the
same derivation, we obtain the symmetric rates in the first hop
\bqa R_p^{(1)}&=&\gamma\left(\frac{a^2\alpha P_1}{1+\alpha
P_1}\right)\label{eq:6(1)}\\
R_c^{(1)}&=&\min\left\{\gamma(\frac{\bar{\alpha}P_1}{\sigma_1^2}),
\frac{1}{2}\gamma\left(\frac{(1+a^2)\bar{\alpha}P_1}{\sigma_1^2}\right)\right\}\label{eq:6(2)}
\eqa where $\alpha\in [0,1]$ and $\sigma_1^2=1+(1+a^2)\alpha P_1$.

The symmetric rates in the second hop under DPC is \bqa
R_{p,DPC}^{(2)}&=&\gamma\left(\frac{b^2\beta P_2}{1+\beta P_2}\right)\label{eq:6(3)}\\
R_{c,DPC}^{(2)}&=&\frac{1}{2}\gamma\left(\frac{(b^2-1)^2\bar{\beta}^2P_2^2}{\sigma_2^4}+\frac{2(1+b^2)\bar{\beta}P_2}{\sigma_2^2}\right)\label{eq:6(4)}
\eqa where $\beta\in [0,1]$ and $\sigma_2^2=1+(1+b^2)\beta P_2$.

The symmetric rates in the second hop under MAC is \bqa\nn
R_{p,MAC}^{(2)}\!\!\!\!&=&\!\!\!\!\max_{\alpha}\left\{\min\left[\gamma\left(\frac{\bar{\alpha}\beta
P_2}{\sigma_3^2}\right),
\frac{1}{2}\gamma\left(\frac{(1+b^2)\bar{\alpha}\beta
P_2}{\sigma_3^2}\right)\right]\right.\\
&&+\left.\gamma\left(\frac{b^2\alpha\beta P_2}{1+\alpha\beta P_2}\right)\right\}\label{eq:6(5)}\\
R_{c,MAC}^{(2)}\!\!\!\!&=&\!\!\!\!\frac{1}{2}\gamma\left(\frac{(1+b)^2\bar{\beta}P_2}{1+(1+b^2)\beta
P_2}\right)\label{eq:6(6)}
 \eqa where $\sigma_3^2=1+(1+b^2)\alpha\beta P_2$ and $\alpha, \beta \in
[0,1]$.

\begin{theorem}\label{thm:2}
The solution to the following optimization problem is achievable
for the two-hop network when $a>1$ and $b>1$:
\bqa R&=&\max_{\alpha,\beta\in [0,1]} R_p+R_c \label{eq:6}\\
&&\mbox{s.t.} (R_p,R_c)\in \Rmat(R_p^{(1)},R_c^{(1)})\cap
\Rmat(R_p^{(2)},R_c^{(2)})\eqa where $R_p^{(1)}$ and $R_c^{(1)}$
are given in (\ref{eq:6(1)})-(\ref{eq:6(2)}).
$\Rmat(R_p^{(2)},R_c^{(2)})$ is defined as the convex closure of
the union of $\Rmat(R_{p,DPC}^{(2)}, R_{c,DPC}^{(2)})$ and
$\Rmat(R_{p,MAC}^{(2)}, R_{c,MAC}^{(2)})$, where $R_{p,DPC}^{(2)}$
and $R_{c,DPC}^{(2)}$ are given in
(\ref{eq:6(3)})-(\ref{eq:6(4)}), and $R_{p,MAC}^{(2)}$ and
$R_{c,MAC}^{(2)}$ are given in (\ref{eq:6(5)})-(\ref{eq:6(6)}).
\end{theorem}

Note that when $\alpha=\beta=0$, let
$\Rmat(R_p^{(2)},R_c^{(2)})=\Rmat(R_{p,DPC}^{(2)},
R_{c,DPC}^{(2)})$, the rate $R$ defined in (\ref{eq:6}) reduces to
that of (\ref{eq:5}). Since $\Rmat(R_p^{(2)},R_c^{(2)})$ is always
a superset of $\Rmat(R_{p,DPC}^{(2)}, R_{c,DPC}^{(2)})$, the
achievable rate (\ref{eq:5}) is always a subset of (\ref{eq:6}).

\subsection{$0<a<1, b>1$}\label{subsection:case3}
For the first hop, it is a weak interference channel, the
transmission strategy is the same as case \ref{subsection:case1}:
the Han-Kobayashi scheme. Thus, the symmetric achievable rate is
$(R_p^{(1)}, R_c^{(1)})$ given in (\ref{eq:1})-(\ref{eq:2}).

For the second hop, we can still use DPC scheme, thus yielding
rates $(R_{p,DPC}^{(2)}, R_{c,DPC}^{(2)})$ given in
(\ref{eq:3})-(\ref{eq:4}). Now consider the MAC scheme. From the
standard result of strong interference channel, the capacity is
achieved when both user's messages are decoded by both receivers,
as in the case of compound MAC. Thus, for the MAC scheme proposed
in section \ref{subsection:case1}, we should modify it by letting
both receivers decode all the messages, both private and common,
instead of further splitting the private message. As such, we
should set $\alpha=0$ in (\ref{eq:3(2)})-(\ref{eq:4(2)}). Also
notice that $b>1$, the symmetric achievable rates for the MAC
scheme become \bqa R_{p,MAC}^{(2)}&=&\min\left\{\gamma(\beta P_2),
\frac{1}{2}\gamma((1+b^2)\beta P_2)\right\}\label{eq:6(7)}\\
R_{c,MAC}^{(2)}&=&\frac{1}{2}\gamma\left(\frac{(1+b)^2\bar{\beta}P_2}{1+(1+b^2)\beta
P_2}\right)\label{eq:6(8)}
 \eqa

Therefore, for the case $0<a<1, b>1$, the symmetric achievable
rate for the two hop network has the same form of that in Theorem
\ref{thm:1}, except that $R_{p,MAC}^{(2)}$ and $R_{c,MAC}^{(2)}$
are given in (\ref{eq:6(7)})-(\ref{eq:6(8)}).

\subsection{$a>1, 0<b<1$}
If we stick to the roles of the two relays, for the first hop, the
two relays should decode both users' messages; for the second hop,
we apply DPC scheme for the weak interference channel. However,
similar to case \ref{subsection:case2}, it can be verified that
this scheme is easily outperformed if we switch the role of the
two relays. Consequently, the first hop becomes a weak
interference channel and the second hop becomes a strong
interference channel. We can directly apply the results from case
\ref{subsection:case3}, with only minor modifications: change the
channel gains $a$ and $b$ into $\frac{1}{a}$ and $\frac{1}{b}$
respectively, and change the variance of noise $Z_1$ and $Z_2$ to
$\frac{1}{a^2}$, and change the variance of noise $Z_3$ and $Z_4$
to $\frac{1}{b^2}$. Thus, the total symmetric rate of the two hop
network becomes the same form of that in Theorem \ref{thm:2}
except that $R_{p,MAC}^{(2)}$ and $R_{c,MAC}^{(2)}$ are given in
(\ref{eq:6(9)})-(\ref{eq:6(10)}). \bqa
R_{p,MAC}^{(2)}&=&\min\left\{\gamma(b^2\beta P_2),
\frac{1}{2}\gamma((1+b^2)\beta P_2)\right\}\label{eq:6(9)}\\
R_{c,MAC}^{(2)}&=&\frac{1}{2}\gamma\left(\frac{(1+b)^2\bar{\beta}P_2}{1+(1+b^2)\beta
P_2}\right)\label{eq:6(10)}
 \eqa

For the second hop, the DPC scheme and the MAC scheme are both
needed for all the parameter regimes. Neither scheme can dominate
the other.

From the previous analysis of the four parameter regimes, we have
the following theorem.
\begin{theorem}\label{thm:role_switching}
For the two hop interference network with the transmission scheme
of decode and forward relaying, if the first hop has weak
interference, one should apply the HK scheme directly; if the
first hop has strong interference, it is always favorable to
convert it into a weak interference channel by switching the roles
of the two relays, as in Fig. \ref{fig:transform}, and then apply
the HK scheme. In other words, with strong interference in the
first hop, rate splitting after role switching of the two relays
can always achieve a rate region no smaller than that achieved by
both relays decoding all the messages without role switching.
\end{theorem}

\begin{proof} If the two relays do not switch roles, for strong
interference in the first hop, the optimal scheme is for both the
two relays to decode all the messages of the two users. Then, the
optimal scheme for the second hop is to use DPC scheme as in the
MIMO broadcast channel. However, these schemes are special cases
of the transmission schemes if we switch the roles of the two
relays and apply the HK scheme to the first hop(simply by
allocating zero power to the private messages after rate
splitting). Therefore, role exchange for the two relay nodes is
always preferred for strong interference in the first hop.
\end{proof}

\subsection{Half Duplex}
If the transmission is conducted in the half duplex fashion, the
two relays cannot receive and transmit at the same time. In this
case, the transmission in the two hops cannot proceed
simultaneously. When transmitting in the first hop, the relays are
in the listening mode and the two users $T_1, T_2$ transmit their
messages with $N_1$ channel uses to the relays. In the second hop,
after decoding the received messages, the two relays $R_1, R_2$
transmit with $N_2$ channel uses to the two destinations $D_1,
D_2$. Thus, the transmission schemes discussed for the full duplex
case can be directly applied to the half duplex case, only with
the overall rates reduced due to the extra channel uses needed.

Following the schemes proposed for the full duplex mode, we always
do rate splitting and transmit private as well as common messages
in the first hop. Thus, both private and common messages should be
successfully delivered to the destinations in the second hop,
which yields: \bqa R_p^{(1)}N_1\leq R_p^{(2)}N_2\\
R_c^{(1)}N_1\leq R_c^{(2)}N_2 \eqa The minimum channel uses needed
in the second hop is \bqa
N_2=N_1\cdot\max\left(\frac{R_p^{(1)}}{R_p^{(2)}},
\frac{R_c^{(1)}}{R_c^{(2)}}\right) \eqa Therefore, the overall
rate achieved for the entire system is \bqa
R=\frac{(R_p^{(1)}+R_c^{(1)})N_1}{N_1+N_2}=\frac{R_p^{(1)}+R_c^{(1)}}{1+\max\left(\frac{R_p^{(1)}}{R_p^{(2)}},
\frac{R_c^{(1)}}{R_c^{(2)}}\right)}\label{eq:7} \eqa
\begin{theorem}
$R^*=\max R$ is the achievable symmetric rate ($R_1=R_2=R^*$) in
the half duplex two-hop interference network, where $R$ is defined
in (\ref{eq:7}).
\end{theorem}

\section{Amplify and Forward}
In this section, we focus on the transmission rates achieved by
amplify and forward relaying. We show that this scheme can
outperform decode and forward relaying under certain conditions.

For amplify and forward relaying, we still focus on the symmetric
channel model as defined in (\ref{eq:model5})-(\ref{eq:model8}).

\subsection{In-phase Relaying} \label{subsection:in-phase}
We first analyze the achievable rates for the so-called in-phase
transmission, where the two relays simply scale their received
signals with the same polarity. This is the usual amplify and
forward scheme and we emphasize in-phase here to contrast with the
out-of-phase approach described later. In the first hop, the
received signals at the
relays are \bqa Y_1&=&X_1+aX_2+Z_1\\
Y_2&=&aX_1+X_2+Z_2 \eqa If they use the full power for amplifying
in
the second hop, we have \bqa X_3&=&cY_1\\
X_4&=&cY_2 \eqa where $c=\sqrt{\frac{P_2}{(1+a^2)P_1+1}}$. Therefore \bqa Y_3&=&cY_1+bcY_2+Z_3\\
Y_4&=&bcY_1+cY_2+Z_4 \eqa after scaling, \bqa
\!\!\!Y_3^{'}&=&(1+ab)X_1+(a+b)X_2+Z_1+bZ_2+Z_3/c\label{eq:model9}\\
\!\!\!Y_4^{'}&=&(a+b)X_1+(1+ab)X_2+bZ_1+Z_2+Z_4/c\label{eq:model10}
\eqa Due to the fact that receivers $D_1$ and $D_2$ do not talk to
each other, we can modify the model in
(\ref{eq:model9})-(\ref{eq:model10}) to the following one without
affecting its capacity region: \bqa
Y_3&=&(1+ab)X_1+(a+b)X_2+Z_3^{'}\label{eq:model11}\\
Y_4&=&(a+b)X_1+(1+ab)X_2+Z_4^{'}\label{eq:model12} \eqa where
$Z_3, Z_4\sim N(0, 1+b^2+1/c^2)$ are independent variables.

\subsubsection{Strong Interference}
It is clear that the model in
(\ref{eq:model11})-(\ref{eq:model12}) will be a strong
interference channel if $a+b>1+ab$, i.e., \bqa
\{a<1,b>1\}\hspace{.3cm} \mbox{or} \hspace{.3cm}\{a>1, b<1\} \eqa
For this model, the optimal scheme is for the two receivers to
decode both users messages, and the capacity region is known as
\bqa R_1&\leq&\gamma\left(\frac{(1+ab)^2P_1}{1+b^2+1/c^2}\right)\\
R_2&\leq&\gamma\left(\frac{(1+ab)^2P_1}{1+b^2+1/c^2}\right)\\
R_1+R_2&\leq&\gamma\left(\frac{((1+ab)^2+(a+b)^2)P_1}{1+b^2+1/c^2}\right)
\eqa Thus, the symmetric achievable rate $(R_1=R_2=R)$ is \beq
\begin{array}{ll}
R=\min\left\{\gamma\left(\frac{(1+ab)^2P_1}{1+b^2+1/c^2}\right),
\frac{1}{2}\gamma\left(\frac{((1+ab)^2+(a+b)^2)P_1}{1+b^2+1/c^2}\right)\right\}\end{array}\label{eq:9}
\eeq

\subsubsection{Weak Interference}
On the other hand, if $a+b<1+ab$, i.e., \bqa
\{a>1,b>1\}\hspace{.3cm} \mbox{or} \hspace{.3cm}\{a<1, b<1\} \eqa
the model (\ref{eq:model11})-(\ref{eq:model12}) becomes a weak
interference channel, for which the Han-Kobayashi's scheme is the
best known scheme. Similar to the analysis in
\ref{subsection:case1}, the symmetric private rate and common rate
are \bqa R_p&\leq& \gamma\left(\frac{(1+ab)^2\alpha
P_1}{(a+b)^2\alpha
P_1+b^2+1+1/c^2}\right)\label{eq:10}\\
R_c&\leq&\min\left\{\gamma\left(\frac{(a+b)^2\bar{\alpha}P_1}{\sigma_1^2}\right),\frac{1}{2}\gamma\left(\frac{\sigma_2^2}{\sigma_1^2}\right)\right\}
\label{eq:11}\eqa where $\sigma_1^2=((1+ab)^2+(a+b)^2)\alpha
P_1+b^2+1+1/c^2$ and
$\sigma_2^2=((1+ab)^2+(a+b)^2)\bar{\alpha}P_1$. The symmetric rate
for the whole system is \bqa R=\max_{\alpha\in [0,1]}R_p+R_c. \eqa

It is interesting to note that for the method of amplify and
forward relaying, the analysis also shows the four parameter
regimes can actually be divided into two categories, in the sense
of transmission and decoding schemes, where $(a<1, b<1)$ and
$(a>1, b>1)$ belong to one category, and $(a<1, b>1)$ and $(a>1,
b<1)$ belong to the other category. This coincides with the
analysis of the decode and forward relaying in the previous
section.

\subsection{Out-of-phase Relaying}
Besides in-phase relaying, the two relays can also purposely make
the relayed signal out of phase by exactly $180^o$, i.e., change
the sign of the relay output. We show in this subsection that this
scheme can have very nice performance under certain conditions.

Again, by using full power at the two relays and making the
relayed signals out of phase by $180^o$, we have \bqa X_3&=&-cY_1\\
X_4&=&cY_2 \eqa where $c=\sqrt{\frac{P_2}{(1+a^2)P_1+1}}$. Therefore, \bqa Y_3&=&-cY_1+bcY_2+Z_3\\
Y_4&=&-bcY_1+cY_2+Z_4 \eqa which, after scaling, is \bqa
\!\!\!Y_3^{'}&=&(ab-1)X_1+(b-a)X_2-Z_1+bZ_2+Z_3/c\label{eq:model13}\\
\!\!\!Y_4^{'}&=&(a-b)X_1-(ab-1)X_2-bZ_1+Z_2+Z_4/c\label{eq:model14}
 \eqa
Since $D_1$ and $D_2$ cannot talk to each other, we can modify the
model (\ref{eq:model13})-(\ref{eq:model14}) to the following model
with the same capacity region: \bqa
Y_3&=&(ab-1)X_1+(b-a)X_2+Z_3^{'}\label{eq:model15}\\
Y_4&=&(a-b)X_1+(1-ab)X_2+Z_4^{'}\label{eq:model16} \eqa where
$Z_3, Z_4\sim N(0, 1+b^2+1/c^2)$ are independent random noises.

\subsubsection{Strong Interference}
For model (\ref{eq:model15})-(\ref{eq:model16}), this becomes a
strong interference channel if $|ab-1|<|b-a|$, i.e., \bqa
\{a<1,b>1\}\hspace{.3cm} \mbox{or} \hspace{.3cm}\{a>1, b<1\}.
\label{condition:1}\eqa This is exactly the same condition as the
strong interference case in section \ref{subsection:in-phase}.
Similar to the analysis of \ref{subsection:in-phase}, we can
express the symmetric rate for the strong interference case as
\beq
\begin{array}{ll}\label{eq:8}
R=\min\left\{\gamma\left(\frac{(1-ab)^2P_1}{1+b^2+1/c^2}\right),
\frac{1}{2}\gamma\left(\frac{((1-ab)^2+(a-b)^2)P_1}{1+b^2+1/c^2}\right)\right\}.
\end{array}\eeq
Obviously, the rate in (\ref{eq:8}) is less than that in
(\ref{eq:9}). Thus, for amplify and forward relaying, under
condition (\ref{condition:1}), we should employ in-phase relaying
at the two relays.

\subsubsection{Weak Interference}
When $|ab-1|>|b-a|$, the model
(\ref{eq:model15})-(\ref{eq:model16}) becomes a weak interference
channel, i.e., \bqa \{a>1,b>1\}\hspace{.3cm} \mbox{or}
\hspace{.3cm}\{a<1, b<1\} \eqa which is also consistent with the
condition of the weak interference case in section
\ref{subsection:in-phase}. Using Han-Kobayashi's scheme, we get
the symmetric private rate and common rate \bqa R_p&\leq&
\gamma\left(\frac{(1-ab)^2\alpha P_1}{(a-b)^2\alpha
P_1+b^2+1+1/c^2}\right)\label{eq:12}\\
R_c&\leq&\min\left\{\gamma\left(\frac{(1-ab)^2\bar{\alpha}P_1}{\sigma_1^2}\right),\frac{1}{2}\gamma\left(\frac{\sigma_2^2}{\sigma_1^2}\right)\right\}
\label{eq:13}\eqa where $\sigma_1^2=((1-ab)^2+(a-b)^2)\alpha
P_1+b^2+1+1/c^2$ and
$\sigma_2^2=((1-ab)^2+(a-b)^2)\bar{\alpha}P_1$.

Comparing rates of (\ref{eq:12})-(\ref{eq:13}) and that of
(\ref{eq:10})-(\ref{eq:11}), it can be easily verified that when
$a=b$ and $ab>>1$(or $ab<<1$), (\ref{eq:12})-(\ref{eq:13}) will
outperform (\ref{eq:10})-(\ref{eq:11}).

If we consider this two hop interference channel network as two
water pipes cascaded with each hop as one pipe, it is very nature
to think the total throughput of the entire system should be
bounded by the capacities of both pipes (e.g., min cut), which is
exactly the case for the decode and forward relaying. However, for
the amplify and forward relaying, we show that this natural
analogy is not valid, i.e., the total throughput can be larger
than the capacities of both ``pipes".

If $a=b$, the model (\ref{eq:model15})-(\ref{eq:model16}) becomes
two parallel AWGN channels and the rates for both channels are the
same: \bqa R=\gamma\left(\frac{(1-a^2)^2P_1}{1+a^2+1/c^2}\right)
=\gamma\left(\frac{(1-a^2)^2P_1P_2}{(1+a^2)(P_1+P_2)+1}\right)\label{eq:14}
\eqa

If $a=b>1$, according to Theorem \ref{thm:role_switching}, the
capacity of each of the two hops is always less than or equal to
that of the transformed channel where we switch the roles of the
two relays, thus converting the strong interference into weak
interference for both hops. Therefore, without loss of generality,
we only consider the case when $a=b<1$. For the interference
channel of the first hop, by
\cite{Shang&Kramer&Chen:09IT,Motahari&Khandani:09IT,Annapureddy&Veeravalli:09IT},
the channel has ``noisy interference" when \bqa a(a^2P_1+1)\leq
\frac{1}{2} \mbox{i.e.,} P_1\leq
\frac{1}{a^2}\left(\frac{1}{2a}-1\right)\label{eq:16} \eqa Under
noisy interference, we know the sum rate capacity of the channel
\cite{Shang&Kramer&Chen:09IT,Motahari&Khandani:09IT,Annapureddy&Veeravalli:09IT},
which is achieved by treating the other user's signal as pure
noise. Thus, the corresponding symmetric capacity is \bqa C_1=
\gamma\left(\frac{P_1}{1+a^2P_1}\right)\label{eq:15} \eqa


If $P_2=P_1$, the symmetric capacity of the second hop is also
$C_2=C_1=\gamma\left(\frac{P_1}{1+a^2P_1}\right)$. In order for
the rate (\ref{eq:14}) to exceed the capacity of both hops for
$P_1=P_2$, i.e., \bqa
\gamma\left(\frac{(1-a^2)^2P_1P_2}{(1+a^2)(P_1+P_2)+1}\right)>\gamma\left(\frac{P_1}{1+a^2P_1}\right)
\eqa we need to satisfy \bqa
P_1>\frac{1+4a^2-a^4+\sqrt{(1+4a^2-a^4)^2+4a^2(1-a^2)^2}}{2a^2(1-a^2)^2}
\eqa Combining (\ref{eq:16}), we get \beq
\begin{array}{ll}
\frac{1+4a^2-a^4+\sqrt{(1+4a^2-a^4)^2+4a^2(1-a^2)^2}}{2a^2(1-a^2)^2}<P_1<
\frac{1}{a^2}\left(\frac{1}{2a}-1\right)\label{eq:17}
\end{array} \eeq
We can easily check that when $a$ is close to 0, the lower bound
of (\ref{eq:17}) is $O(\frac{1}{a^2})$ and the upper  bound of
(\ref{eq:17}) is $O(\frac{1}{a^3})$, which indicates that when $a$
is close to 0, such $P_1$ does exist. For example, when $a=0.15$,
the bound in (\ref{eq:17}) becomes $51.6<P_1<103.7$.

The above example is for $a=b<1$. Similarly, for $a=b>1$, due to
the previous analysis that these two cases are essentially
identical (by switching the roles of the two relays), it can be
verified that when $a=b>>1$, the transmission rate for the whole
system can also exceed the capacity of each individual
interference channel. The details are omitted here.

Although the above results are obtained for $a=b$, we comment that
even when $a\neq b$ but are close, one can still find parameter
regimes for which the out-of-phase scheme is favored, i.e., has a
larger symmetric rate.

\section{Numerical Examples}
For both decode-and-forward relaying and amplify-and-forward
relaying, when the the first hop has strong interference, i.e.,
$a>1$, it is always preferred to switch the roles of the two
relays and convert the channel into a weak interference channel.
Without loss of generality, we only focus on the weak interference
case of the first hop, i.e., $a<1$. First, we compare the effect
of the two schemes in the second hop, namely DPC scheme and MAC
scheme, under different channel parameters for the
decode-and-forward relaying.
\begin{figure}[htp]
\begin{tabular}{cc}
\leavevmode \epsfxsize=1.8in \epsfysize=1.3in
\epsfbox{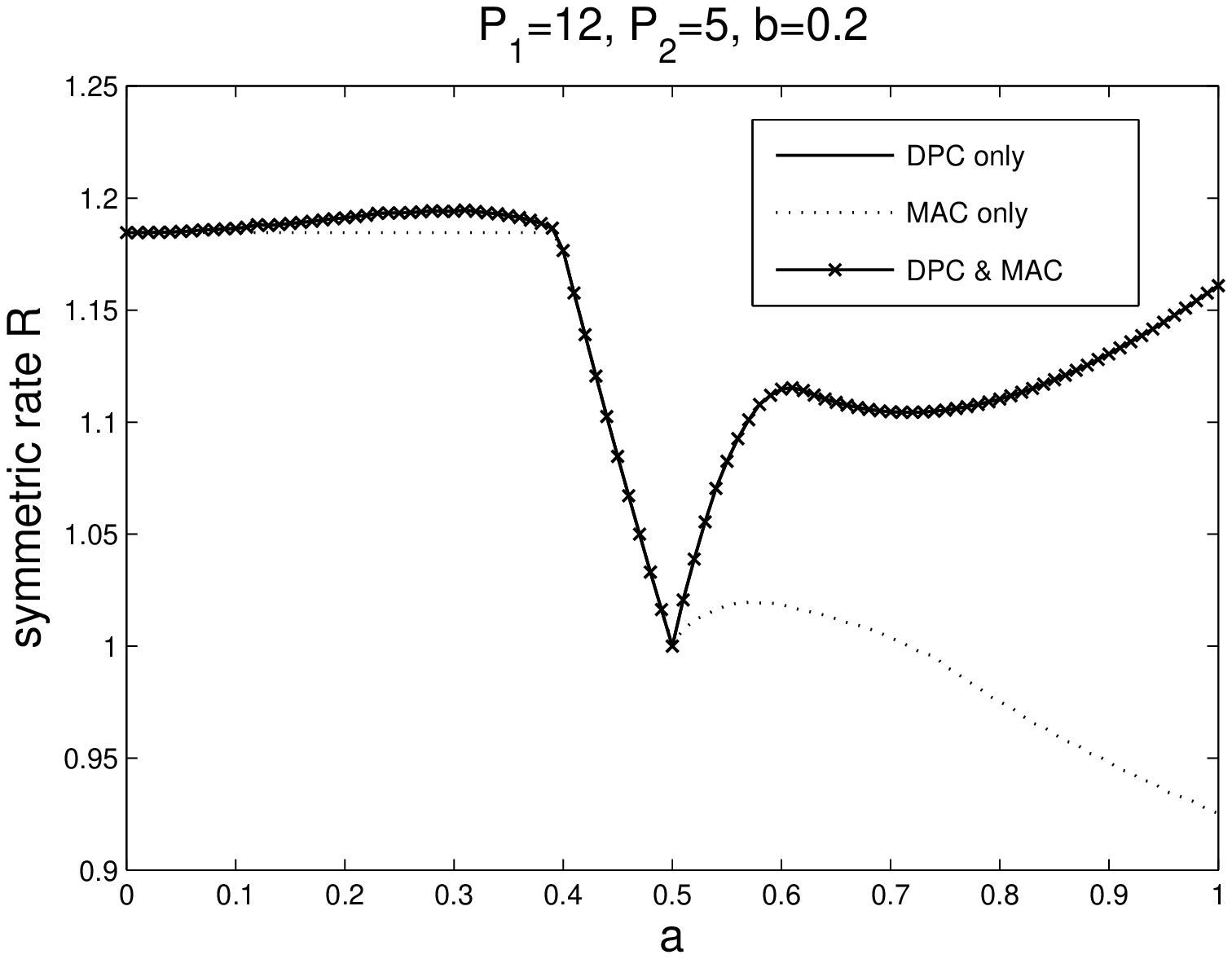}& \leavevmode \epsfxsize=1.8in
\epsfysize=1.3in \epsfbox{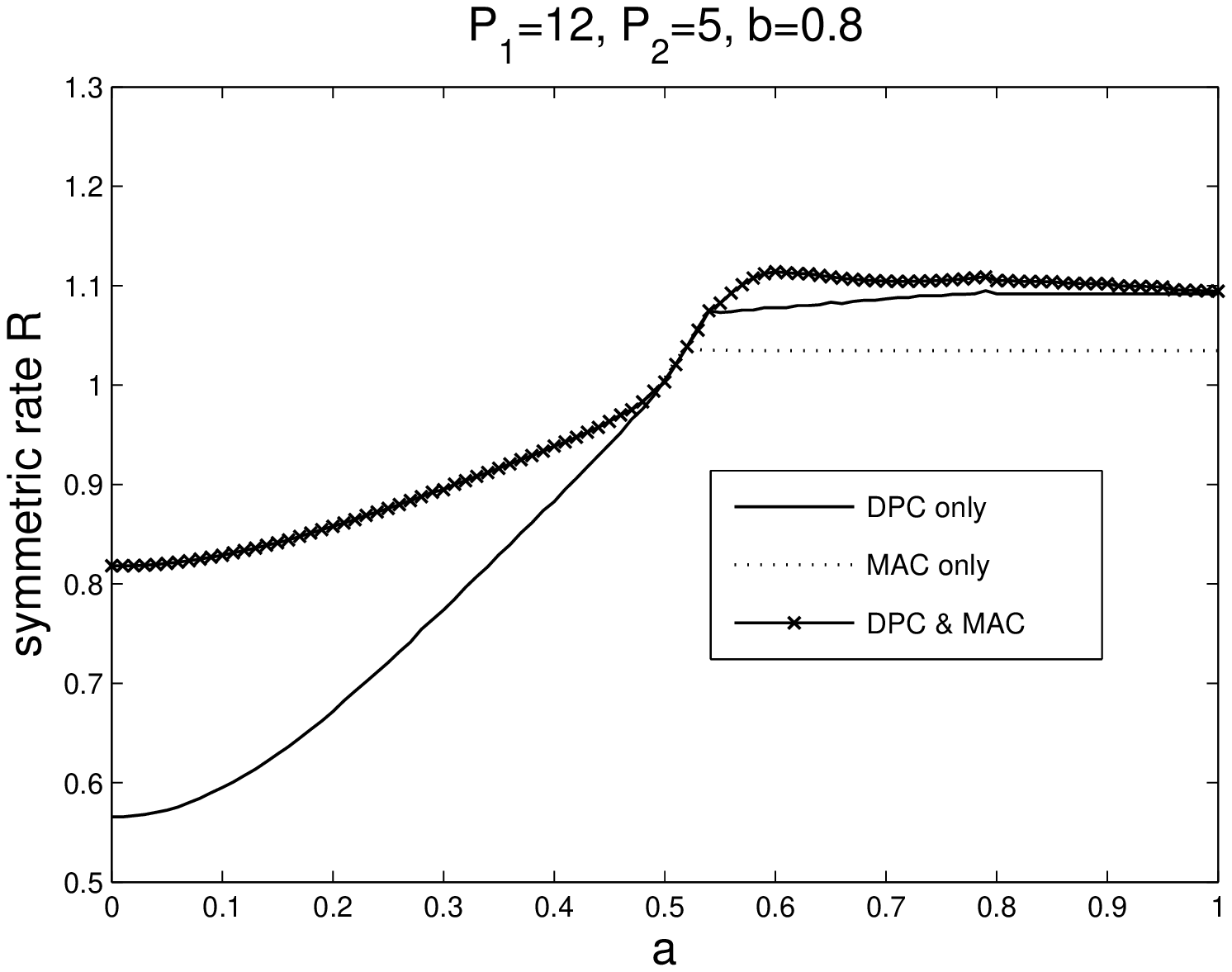}\\
(a)&(b)\\
\leavevmode \epsfxsize=1.8in \epsfysize=1.3in \epsfbox{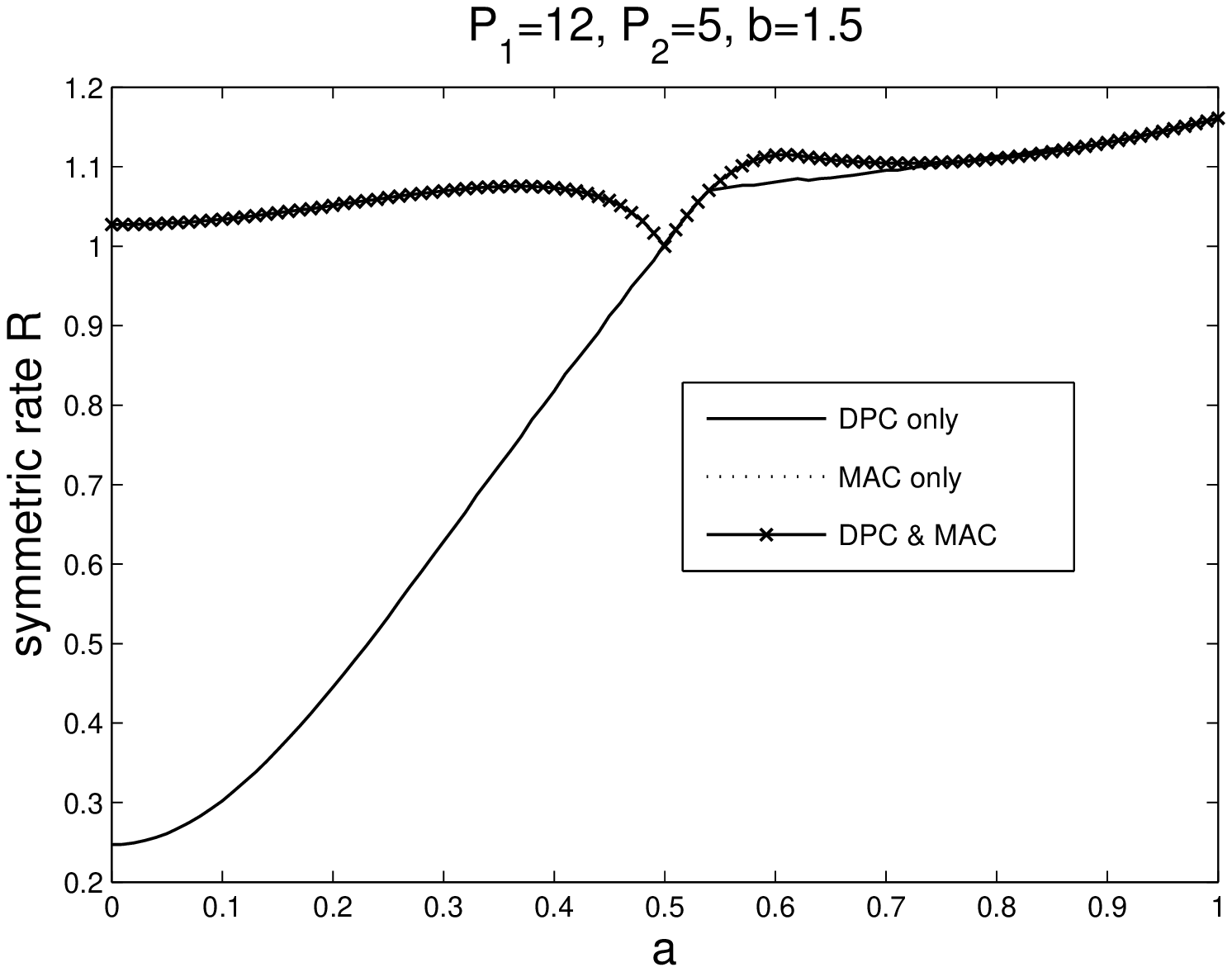}\\
(c)
\end{tabular}
\caption{\label{fig:DPC_vs_MAC} Comparison of DPC scheme and MAC
scheme in the second hop for the decode-and-forward relaying.}
\end{figure}

Fig. \ref{fig:DPC_vs_MAC} (a) shows that when the interference
gain of the second hop $b$ is very small, the DPC scheme is
dominating for $a\in [0,1]$ and the symmetric rate for combining
DPC and MAC will coincide with that of DPC scheme only. The
difference between DPC and MAC becomes dramatic when $a>0.5$. That
is because in this regime, the HK scheme will produce significant
amount of common information in the first hop, and MAC scheme
requires the common information to be decoded by both receivers,
which negatively affects the total rate since $b$ is small at the
second hop. However, when $b$ gets larger, as shown in (b), MAC
scheme will beat DPC for $a<0.5$ but will be outperformed by DPC
for $a>0.5$. Since for $a<0.5$, there is significant amount of
private messages produced by HK scheme in the first hop, which
will be treated as noise in the DPC scheme, but will be partially
decoded in the MAC scheme, thus MAC will perform better. However
for $a>0.5$, the common messages from the first hop dominates.
Since DPC scheme can cancel the interference effect of other
user's common messages, this advantage beats the MAC scheme where
the common messages need to be decoded by both receivers when $b$
is not strong enough $(b=0.8)$. Note that the combination of DPC
and MAC will outperform both of the individual schemes for $a>0.5$
because of the time sharing effect of the two rate regions. When
$b$ gets strong enough as in (c), the MAC scheme will far
outperform DPC when $a$ is small, but will be close to DPC when
$a$ gets larger.

Next, we show in Fig. \ref{fig:DF_vs_AF} the comparison of
decode-and-forward relaying and amplify-and-forward relaying (both
in phase and $180^o$ out of phase) in low SNR regime.

\begin{figure}[htp]
\begin{tabular}{cc}
\leavevmode \epsfxsize=1.8in \epsfysize=1.3in
\epsfbox{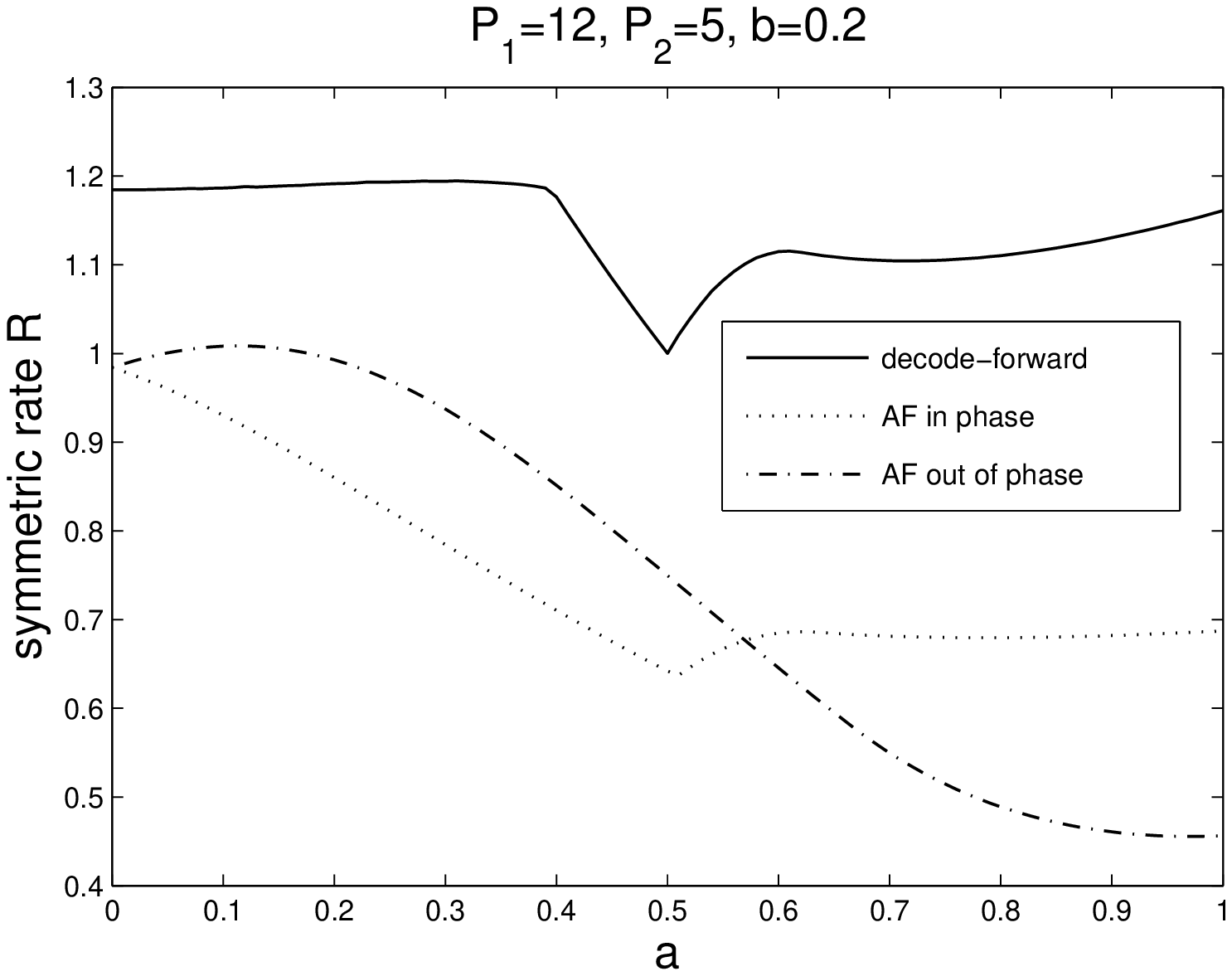}& \leavevmode \epsfxsize=1.8in
\epsfysize=1.3in \epsfbox{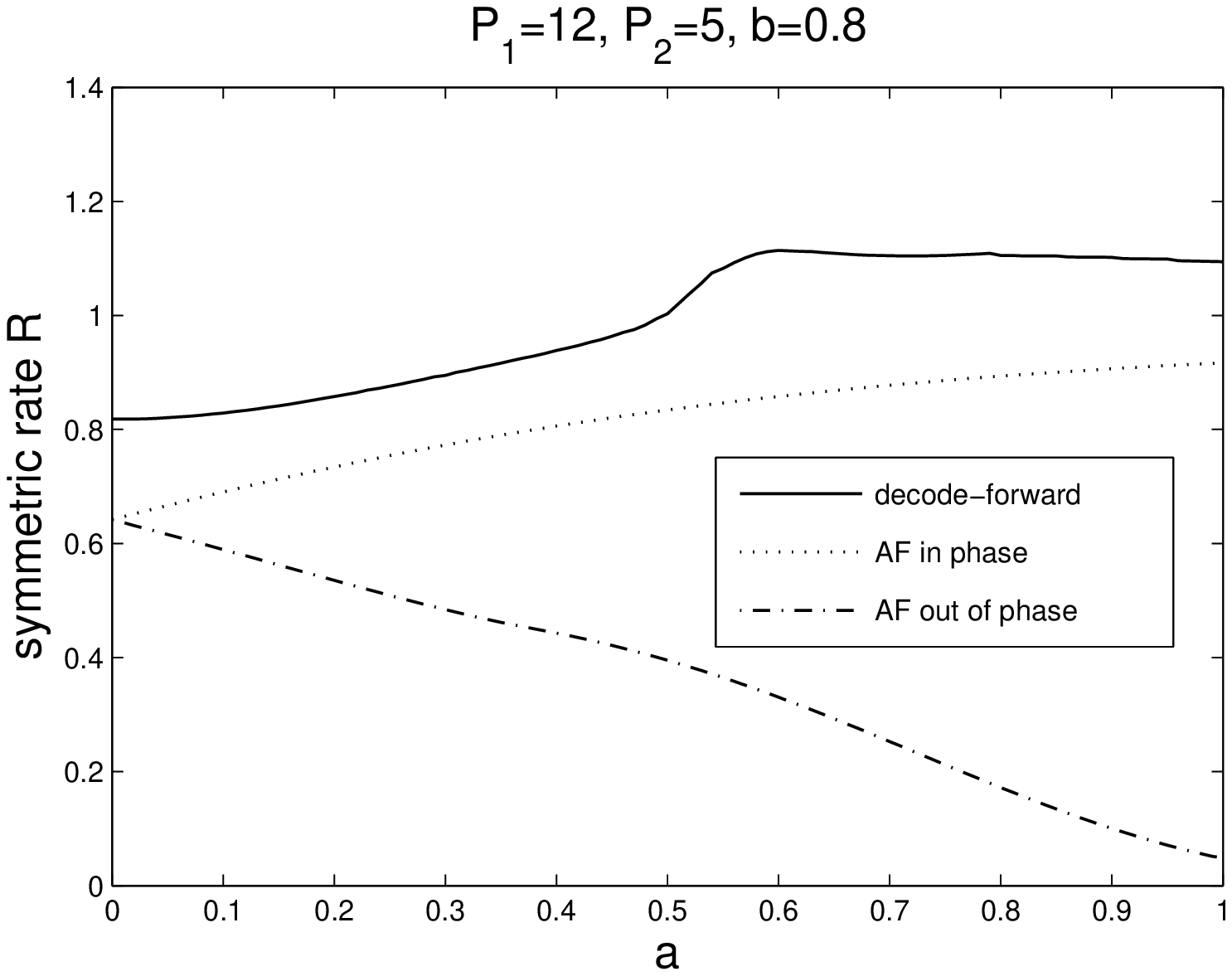}\\
(a)&(b)\\
\leavevmode \epsfxsize=1.8in \epsfysize=1.3in \epsfbox{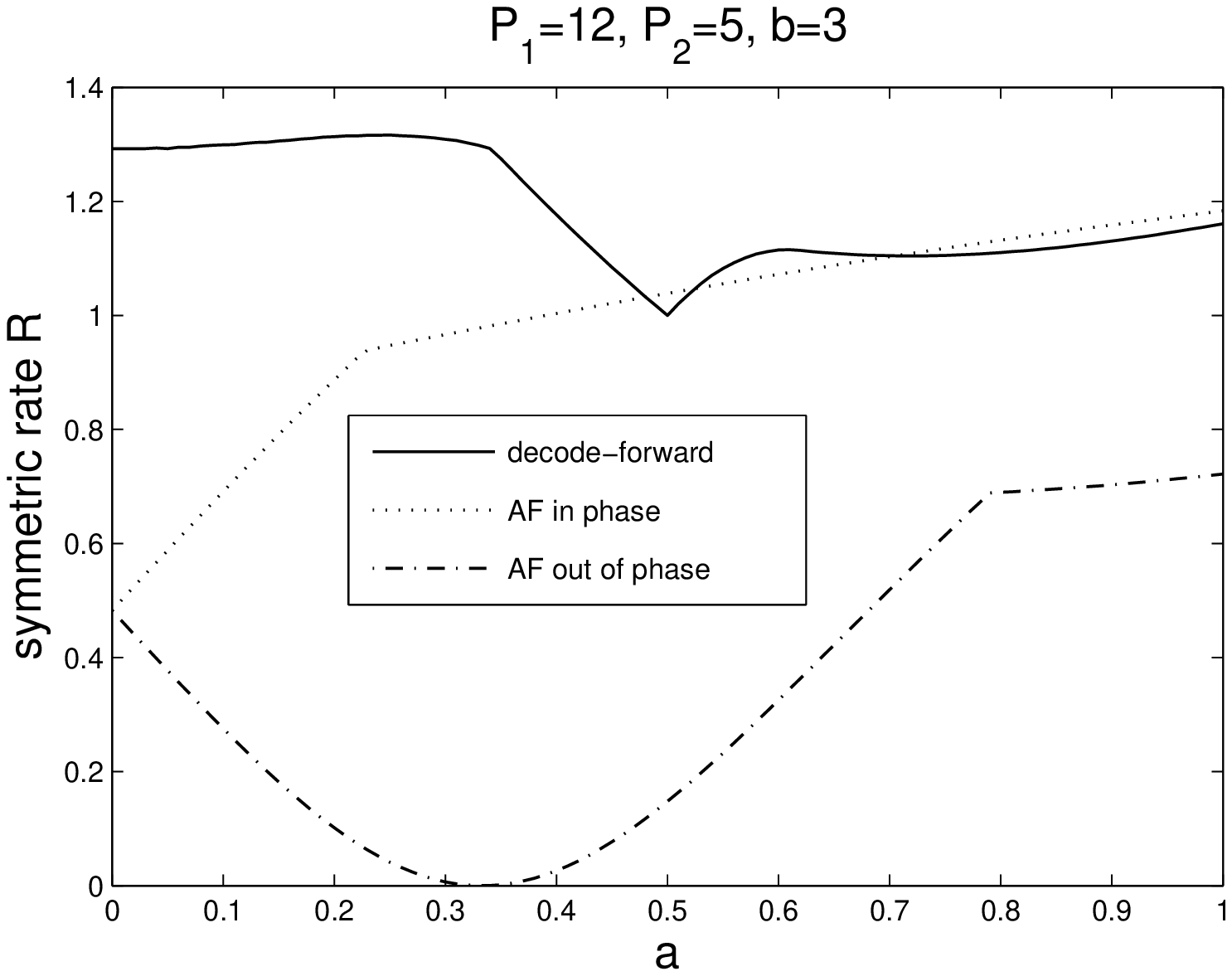}\\
(c)
\end{tabular}
\caption{\label{fig:DF_vs_AF} Comparison of decode-and-forward
relaying and amplify-and-forward relaying in low SNR regime}
\end{figure}
It can be seen that for the low SNR regime, when $b$ is small, the
amplify-and-forward relaying scheme (for both in phase and $180^o$
out of phase) will always be outperformed by the
decode-and-forward scheme, as shown in (a) and (b). When $b$ gets
strong enough, as shown in (c), the in phase amplify-and-forward
relaying may outperform, but not by much, the decode-and-forward
scheme when $a$ is close to 1. In other words, for the low SNR
regime, decode-and-forward scheme is preferred over
amplify-and-forward scheme. However, at high SNR regime, it is a
different story.

\begin{figure}[htp]
\begin{tabular}{cc}
\leavevmode \epsfxsize=1.8in \epsfysize=1.3in
\epsfbox{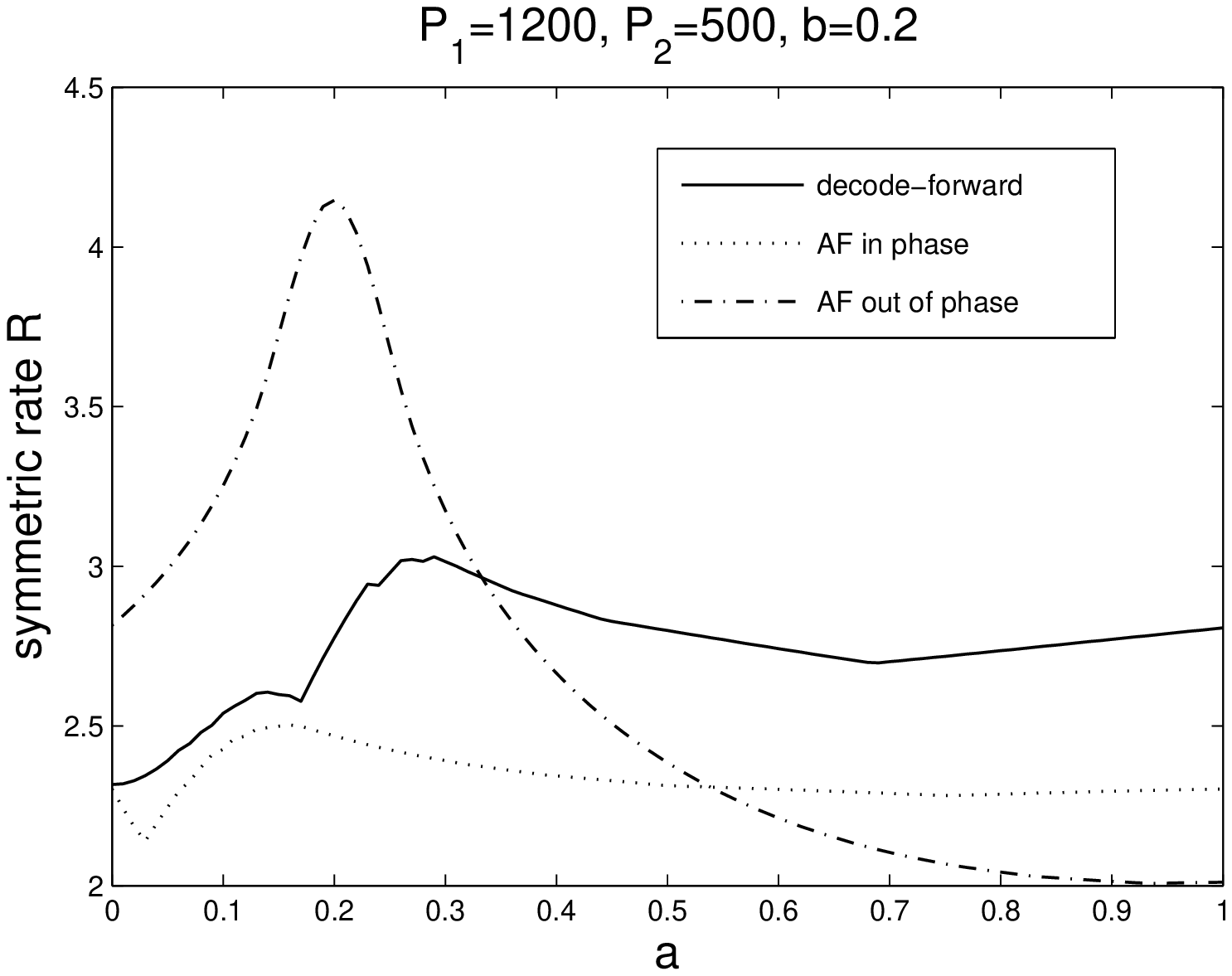} & \leavevmode \epsfxsize=1.8in
\epsfysize=1.3in \epsfbox{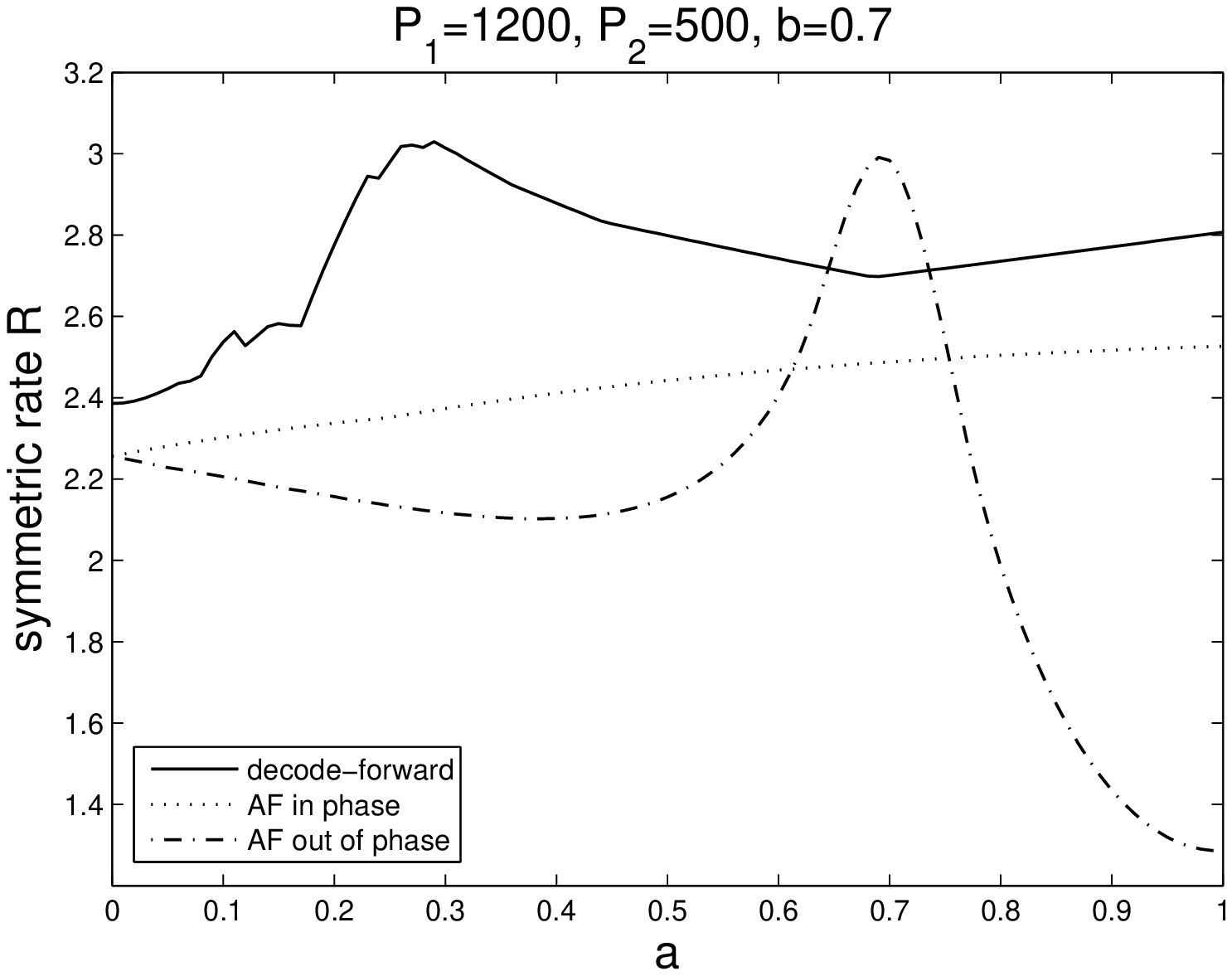}\\
(a)&(b)\\
\leavevmode \epsfxsize=1.8in \epsfysize=1.3in \epsfbox{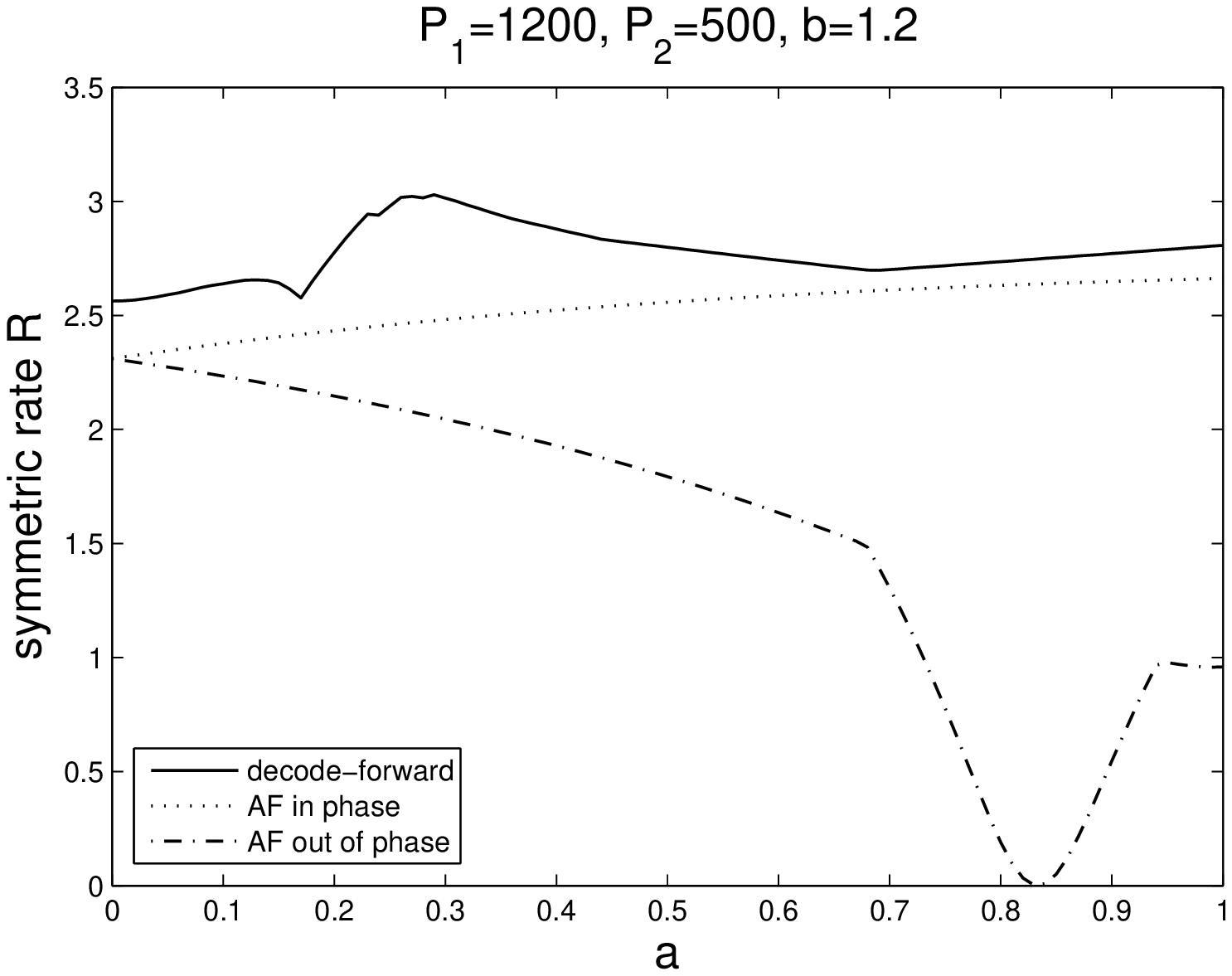}\\
(c)
\end{tabular}
\caption{\label{fig:DF_vs_AF2} Comparison of decode-and-forward
relaying and amplify-and-forward relaying in high SNR regime}
\end{figure}

As shown in Fig. \ref{fig:DF_vs_AF2}, in the high SNR regime, when
$b<1$, the performance of amplify-and-forward relaying with
$180^o$ out of phase is the best when $a$ is close to $b$. This is
because when $a=b$, the channel becomes two parallel AWGN
channels, which has the best performance under the high SNR
regime. However, away from the peak of $a=b$, the
amplify-and-forward relaying with $180^o$ out of phase is still
the worst. When $b\geq 1$, since $a\in [0,1]$, the peak of $a=b$
does not exist any more, thus, the performance of the out of phase
amplify-and-forward relaying becomes the worst for all values of
$a$. In this case, the decode-and-forward scheme remains the best
of all.

\section{Conclusion}
In this paper, we investigated and compared coding schemes for the
two hop interference network under various channel parameters
regimes. Our analysis shows that if the first hop has strong
interference, i.e., $a>1$, it is always beneficial to switch the
roles of the two relays so that the channel is converted to a weak
interference channel with interference gain of $1/a$, and the
strength of the second hop is also changed accordingly.

For the decode-and-forward relaying, the DPC scheme and MAC scheme
are both needed for the second hop. The combination of the two may
sometimes outperform both of the individual schemes due to the
time sharing effect. Generally however, DPC scheme dominates when
$b$ is small and MAC scheme dominates when $b$ is large.

The comparison of decode-and-forward relaying and
amplify-and-forward relaying showed that decode-and-forward
relaying always has better performance except when $a$ is close to
$b$ in the high SNR regime.

\bibliographystyle{C://localtexmf/caoyibib/IEEEbib}
\bibliography{C://localtexmf/caoyibib/Journal,C://localtexmf/caoyibib/Conf,C://localtexmf/caoyibib/Book}

\end{document}

%% file: macro.tex
\newtheorem{theorem}{Theorem}
\newtheorem{example}{Example}

\newtheorem{definition}{Definition}

\def\psfancypar#1#2{\begingroup\def\par{\endgraf\endgroup\lineskiplimit=0pt}
               \setbox2=\hbox{\large\sc #2}
               \newdimen\tmpht \tmpht \ht2 \advance\tmpht by \baselineskip
               \font\hhuge=Times-Bold at \tmpht
               \setbox1=\hbox{{\hhuge #1}}
               \count7=\tmpht \count8=\ht1
               \divide\count8 by 1000 \divide\count7 by \count8 
               \tmpht=.001\tmpht\multiply\tmpht by \count7 
               \font\hhuge=Times-Bold at \tmpht
               \setbox1=\hbox{{\hhuge #1}}
               \noindent
                \hangindent1.05\wd1
               \hangafter=-2 {\hskip-\hangindent
               \lower1\ht1\hbox{\raise1.0\ht2\copy1}%
                \kern-0\wd1}\copy2\lineskiplimit=-1000pt}

\newcommand{\beq}{\begin{equation}}
\newcommand{\eeq}{\end{equation}}
\newcommand{\bqa}{\begin{eqnarray}}
\newcommand{\eqa}{\end{eqnarray}}
\newcommand{\bqn}{\begin{eqnarray*}}
\newcommand{\eqn}{\end{eqnarray*}}
\newcommand{\nn}{\nonumber}

\newcommand{\be}{\begin{enumerate}}
\newcommand{\ee}{\end{enumerate}}
\newcommand{\bi}{\begin{itemize}}
\newcommand{\ei}{\end{itemize}}
\newcommand{\bd}{\begin{description}}
\newcommand{\ed}{\end{description}}
\newcommand{\ba}{\begin{array}}
\newcommand{\ea}{\end{array}}
\newcommand{\bde}{\begin{definition}}
\newcommand{\ede}{\end{definition}}
\newcommand{\bex}{\begin{example}}
\newcommand{\eex}{\end{example}}

\def\boxit#1{\vbox{\hrule\hbox{\vrule\kern3pt
        \vbox{\kern3pt#1\kern3pt}\kern3pt\vrule}\hrule}}

\def\reals{ { {\rm  I \kern-0.15em R }  } }
\def\complex{ {\,{{\rm C} \kern-0.50em \raise0.20ex {  |}}\, }}

\def\0bf{{\bf 0}}
\def\1bf{{\bf 1}}
\def\2bf{{\bf 2}}
\def\3bf{{\bf 3}}
\def\4bf{{\bf 4}}
\def\5bf{{\bf 5}}
\def\6bf{{\bf 6}}
\def\7bf{{\bf 7}}
\def\8bf{{\bf 8}}
\def\9bf{{\bf 9}}

\def\Rbf{{\bf R}}

\def\Rmat{\mathcal{R}}

\def\Rxx{\Rbf_{\ssstyle X\kern-.1em X}}

\let\ssstyle=\scriptscriptstyle

\def\Kout{\setbox1=\hbox{\Huge\bf K}\hbox to
1.05\wd1{\hspace{.05\wd1}
}}
\def\Sout{\setbox1=\hbox{\Huge\bf S}\hbox to 1.05\wd1{\hspace{.05\wd1}
}}

%% file: Allerton.bbl
\begin{thebibliography}{1}

\bibitem{Akyildiz:05CN}
I.~Akyildiz, X.~Wang, and W.~Wang,
\newblock ``Wireless mesh networks: a survey,''
\newblock {\em Computer Networks}, vol. Vol.47, pp. 445--487, 2005.

\bibitem{Simeone_etal:07Allerton}
O.~Simeone, O.Somekh, Y.~Bar-Ness, H.~V. Poor, and S.~Shamai,
\newblock ``{Capacity of linear two-hop mesh networks with rate splitting,
  decode-and-forward relaying and cooperation},''
\newblock in {\em Proc. 45th Annual Allerton Conference on Communication,
  Control and Computing}, 2007.

\bibitem{Han&Kobayashi:81IT}
T.~Han and K.~Kobayashi,
\newblock ``{A new achievable rate region for the interference channel},''
\newblock {\em IEEE Trans. on Information Theory}, vol. IT-27, pp. 49--60, Jan.
  1981.

\bibitem{Thejaswi_etal:07Allerton}
C.~Thejaswi, A.~Bennatan, J.~Zhang, R.~Calderbank, and D.~Cochran,
\newblock ``{Rate-achievability strategies for two-hop interference flows},''
\newblock in {\em Proc. 46th Annual Allerton Conference on Communication,
  Control and Computing}, Sept. 2008.

\bibitem{Shang&Kramer&Chen:09IT}
X.~Shang, G.~Kramer, and B.~Chen,
\newblock ``{A new outer bound and the noisy-interference sum-rate capacity for
  Gaussian interference channels},''
\newblock {\em IEEE Trans. Inform. Theory}, vol. 55, pp. 689--699, Feb. 2009.

\bibitem{Motahari&Khandani:09IT}
A.~S. Motahari and A.~K. Khandani,
\newblock ``{Capacity bounds for the Gaussian interference channel},''
\newblock {\em IEEE Trans. Inform. Theory}, vol. 55, pp. 620--643, Feb. 2009.

\bibitem{Annapureddy&Veeravalli:09IT}
V.~S. Annapureddy and V.~V. Veeravalli,
\newblock ``{Gaussian interference networks: Sum capacity in the low
  interference regime and new outer bounds on the capacity region},''
\newblock {\em submitted to IEEE Trans. Inform. Theory}, Feb. 2008.

\end{thebibliography}
